\documentclass{article}
\usepackage{epsfig, subfigure}
\usepackage{color}
\usepackage{amssymb,amsmath}
\usepackage{graphicx}
\usepackage{a4wide}

\newtheorem{theorem}{Theorem}
\newtheorem{lemma}{Lemma}

\def\Pr{\mathop{\rm Pr}}
\def\area{\mathop{\rm area}}
\def\bd{\mathop{\rm bd}}

\def\peri{\mathop{\rm peri}}
\def\components{\mathop{\rm components}}

\newenvironment{proof}%
{\noindent\emph{Proof}.\hspace{1ex}}%
{\hfill\unitlength=0.18ex%
  \begin{picture}(12,12)
    \put(1,1){\framebox(9,9){}}
    \put(1,4){\framebox(6,6){}}
  \end{picture}\linebreak
}

\begin{document}

\title{
Local Event Boundary Detection with Unreliable Sensors:
Analysis of the Majority Vote Scheme\thanks{Work by C.-S. Shin was supported by National Research Foundation of Korea(NRF) grant funded by the Korea government(MEST) (No. 2011-0002827).}}

\author{
	Peter~Brass\thanks{Dept. of Computer Science, City College, New York, USA \tt{peter@cs.ccny.cuny.edu}} \and 
	Hyeon-Suk Na\thanks{School of Computing, Soongsil University, Seoul, Korea. \tt{hsnaa@ssu.ac.kr}} 
	\and Chan-Su Shin\thanks{Dept. of Electrical Information Engineering, Hankuk University of Foreign Studies, Korea. \tt{cssin@hufs.ac.kr}}
}
\date{}
\maketitle

\begin{abstract}
In this paper we study the identification of an event region $X$ within a
larger region $Y$, in which the sensors are distributed by a Poisson process
of density $\lambda$ to detect this event region, i.e., its boundary. The model of sensor is a 0-1
sensor that decides whether it lies in $X$ or not, and which might be
incorrect with probability $p$. It also collects information on the
0-1 values of the neighbors within some distance $r$ and revises its
decision by the majority vote of these neighbors. In the most general setting, we analyze this simple majority vote scheme
and derive some upper and lower bounds on the expected number of misclassified
sensors. These bounds depend on several sensing parameters of $p$, $r$, and some geometric parameters of the event region $X$. By making some assumptions on the shape of $X$, we prove a significantly improved upper bound on the expected number of misclassified sensors; especially for convex regions with sufficiently round boundary, and we find that the majority vote scheme performs well in the simulation rather than its theoretical upper bound.
\end{abstract}


\section{Introduction}
Suppose we have distributed many sensors in a region, each of which detects
if it rains at that point. We want to obtain a summary: in which sub-region
is it raining? Just listing all the positions at which a raindrop has been
detected is not a helpful answer, first because a long list of positions
is not the answer a user would want on the question ``Where does it rain?,''
but also because each individual answer is subject to random
errors; the detected drop of water could have come from an air-conditioner,
or from children splashing in the water, or numerous other random events,
and in the same way, even if it is raining, the sensor might coincidentally
not catch any drop. We expect a useful answer to the question ``Where does it rain?'' to be some region with a simple structure.

\par
The abstract model underlying this question is as follows: we have a
region $Y$, in which there is a set $S$ of sensors. There is an unknown
region $X \subseteq Y$ in which the event happens. Possible events would
be rain, forest fires, or occurrences of invasive species. We want to detect the event region $X$, most importantly, its boundary, where it is far from the boundary of $Y$ and has a `nice' topology, not being highly irregular, random or fractal. The sensors $s\in S$ are 0-1 sensors who decide whether $s\in X$ or $s\notin X$, making an error with probability $p$ in this measurement, called the measurement error.

\subsection{Simple Majority Voting Scheme}
Our aim is to reduce the error rate in detecting the boundary of $X$
by allowing each sensor to compare its result with those of its neighbors. To reduce the error rate by local communication, we assume that each sensor knows the values measured by all neighboring sensors within distance $r$. The most
straightforward method to use the neighbors' sensing information is
to follow the majority.

This majority vote scheme is as follows: if the sensor has $k$ neighbors and knows its own and those $k$ other measurements, then in its revised decision it just follows the majority of the measurements of its $k$ neighbors, with itself as tie-breaking if necessary.
This scheme was already proposed by Chintalapudi and Govindan~\cite{CG03}
and further in~\cite{KI04}. It does not use the position information of
the neighboring sensors. It was stated in~\cite{KI04} that this scheme
gives a good correction of measurement errors for sensor error $p$ up
to $0.2$. However, we think this observation needs further qualification,
since the situation really depends on the size of voting neighborhood
and several geometric parameters of the boundary of $X$.

\begin{figure}
\centering
\includegraphics[width=11cm]{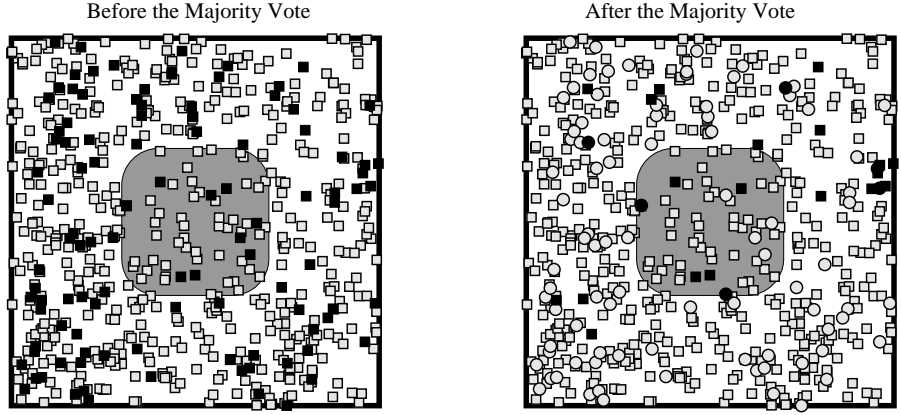}
\caption{$600$ sensors are distributed uniformly in a unit square $Y$, and the gray-colored region is the event region $X$. Among $600$ sensors, $106$ sensors (marked by black boxes in left figure) made wrong initial measurements, but $86$ sensors, $81\%$ of misclassified ones, revised correctly (marked by gray discs in the right figure), and the other $20$ sensors (marked by black boxes) remained still in their wrong classification, but $6$ sensors (marked by black discs) turned from the correct classification to the wrong classification.}
\label{fig:intro}
\end{figure}

\par
Figure~\ref{fig:intro} shows a simulation on the majority vote scheme
with $600$ sensors distributed uniformly in a unit square $Y$. The
event region $X$ is a square of side length $0.4$ with rounded corners of the curvature radius $0.1$. The measurement error probability $p = 0.15$ and the neighborhood radius $r=0.05$. The error correction rate by this scheme is about $76\%$.
From this simulation, one can observe that the total error of event
detection is significantly reduced by the majority vote scheme. In fact,
the error depends on several sensing parameters such as the measurement
error $p$ and the neighborhood radius $r$, and some geometric parameters
of the event region $X$ such as the convexity, the perimeter and the
boundary curvature.

In this paper, we analyze the majority scheme
further, and make explicit and precise the dependency on these parameters. To analyze this, we first need an assumption on the distribution of sensors in the region of interest $Y$. We adapt the most important model; the sensors are randomly distributed by a Poisson process of density $\lambda$, which is independent from the measurement error $p$ that the sensors make. To our best knowledge, this gives the first bounds on the expected number of incorrectly classified sensors in the majority vote scheme with all these parameters, $\lambda$, $p$, $r$, and the shape of $X$.

\par
It should be noticed that our sensors are point sensors; there is
no sensing range, but a yes/no decision about the situation at the sensor.
Many other papers have dealt with continuous-valued sensors, but then the dependence of a sensor's decision on the neighboring values is much less clear. Also, for many practical applications a yes/no decision is ultimately the
desired answer: `does it rain?', `is there a forest fire?', etc.
The distance $r$ within which we compare the sensor values is
not related to the communication distance of the network nodes;
it is a choice made depending on our a-priori knowledge on the
size and shape of $X$, that is, choosing the right radius $r$ is an important aspect.
Finally, our majority vote scheme does not need absolute positions
of the sensors, but it is not efficient in detecting thin and long event regions; to identify such event regions we would need sensors with
known positions with the help of GPS units and by more complicated decision algorithms.

\subsection{Previous Work}
Local event boundary detection problem has been studied in several previous papers.

Chintalapudi and Govindan~\cite{CG03} were the first to analyze the local
event boundary detection problem. They proposed three different types of
algorithms; among them a simple neighborhood counting scheme that does
not use the position information of the neighboring sensors, and a scheme
that finds the optimum line separating in the neighborhood the event-sensors
from the no-event-sensors.
They found by simulation that the separating-line scheme performs best,
but provided no analysis of that scheme, and assumed for the other schemes
that the event boundary is a straight line.
\par
Krishnamachari and Iyengar~\cite{KI04} discussed a model with a similar
counting scheme, not using the neighbor's position information; in their
simulation, however, the sensors are always distributed in a square grid.
Wu et al.~\cite {WCDXLD07} discussed continuous-valued sensors, looking
for threshold events, and proposed several methods
based on comparing a sensor's value with the median of a set of neighbor's
values to identify faulty or boundary sensors. Similar was the discussion
by Jin and Nittel~\cite{JN06}, who used the mean instead of the median.
Ding and Cheng~\cite{DC09} fit a mixture of multivariate gaussian
distributions to the observed sensor values and decided on the base
of that fitted model which sensors are boundary sensors.

\par
Wang et al.~\cite{WGM06} considered a model that can be interpreted as
only no-event sensors being available, e.g., because the event like the fire in the forest destroyed
all sensors in its region; they reconstruct a boundary of the regione
containing sensors, based on neighbor connectivity information without
the neighbor positions. Nowak and Mitra~\cite {NM03} use a non-local
communication model based on a hierarchical partitioning scheme to
identify event boundaries.
\par
A different line of related works contains the fusion of different
information sources for the same event; e.g., combining
the output of multiple classifiers in a pattern-recognition problem.
This has been studied in \cite{LaS97,KiH*98,AlK99,CC01,Ku02,Na03},
but that problem abstracts from the geometric structure which is the
core of our considerations.
\par
Yet another related line of works contains the opinion formation
in social networks: instead of spatially related nodes, we have
persons with friends, and a person might change his opinion
to conform to the majority opinion of his group. This has been
studied both as social dynamics and as abstract process on graphs,
see, e.g., \cite{Ag*88,MuP01,Pe02,Kr02,MuP04,Zo10}.
Majority vote among neighbors has also been studied as
model for some spin systems in physics, see e.g.,
\cite{Ol92,Li*05,Ya*08,Cu*10}.

\subsection{Our Results}

For a set $A$ in the plane, we denote its area, perimeter, boundary and number of components, by $\area(A), \peri(A)$, $\bd(A)$ and $\components(A)$, respectively.

\par
We assume that the set $S$ of sensors is generated by a Poisson process of density $\lambda$ on our region of interest $Y$. Within $Y$, an event happens in the region $X$. Each sensor $s\in S$ makes a 0-1 event detection (or measurement), whether $s \not\in X$ or $s \in X$, which might be incorrect with probability $p$. The sensor errors are independent from each other, and from the Poisson process placing the sensors. Each sensor knows the measurement results of all other sensors within radius $r$ and revises his own measurement based on that information.

\par
The comparison with neighboring sensors gives information only if the sensor has neighboring sensors. The expected number of neighbors in this model is $\lambda\pi r^2$ and thus the probability that a sensor has no neighbors is $e^{-{\lambda\pi r^2}}$, which should be much smaller than the measurement error $p$. The expected number of sensors in $Y$ is $\lambda\area(Y)$, so without correction by neighborhood comparison, the expected number of incorrect sensors, i.e., misclassified sensors, is $\lambda\area(Y)p$. Theorem~\ref{thm:general_outZ} and Theorem~\ref{thm:general_inZ} show that the expected number of misclassified sensors is improved significantly by the majority vote scheme.

\par
Let $Z_r$ be the set of points within distance $r$ to the boundary of $X$. Any sensor in this dubious region $Z_r$ has potentially neighbors inside and outside $X$, in other words, we possibly have both of correct 0- and 1-answers within the same neighborhood, which would lead such sensors to make the wrong decision after the majority vote. Thus the analysis on the expected number of misclassified sensors in $Z_r$ is a key in the majority vote scheme.

In Section~\ref{sec:analysis}, we analyze the majority vote scheme in a most general setting and derive the following bounds on the expected number of misclassified sensors in $Y\setminus Z_r$ and on the expected number of misclassified sensors in $Z_r$.
\begin{theorem}\label{thm:general_outZ}
For $p \leq 1/2$, the expected number of sensors in $Y \setminus Z_r$ that are misclassified by the simple majority rule in the neighborhood of radius $r$ is at most
$$2\lambda \sqrt{p(1-p)}e^{-(1-2\sqrt{p(1-p)})\lambda\pi r^2}\area(Y\setminus Z_r)$$
and at least
$$\frac{\sqrt{p(1-p)}}{4\pi r^2} \left(e^{-(1-2\sqrt{p(1-p)})\lambda\pi r^2}-e^{-\lambda\pi r^2}\right)\area(Y\setminus Z_r).$$
\end{theorem}

\begin{theorem}\label{thm:general_inZ}
For $p \leq 1/2$, the expected number of sensors in $Z_r$ that are misclassified by the simple majority rule in the neighborhood of radius $r$ is at most
$$2\lambda r\peri(X) + \lambda \pi r^2\components(X).$$
There exists some event region $X$ such that the expected number of misclassified sensors in $Z_r$ is at least $\Omega(\lambda r\peri(X))$.
\end{theorem}

The ratio of the upper bound to the lower bound in Theorem~\ref{thm:general_outZ} grows linearly with the expected number of neighbors $\lambda\pi r^2$. However, since $2\sqrt{p(1-p)}\leq 1$ for any $p\geq 0$, they both decrease exponentially with the expected number of neighbors $\lambda\pi r^2$, so the expected number of misclassified sensors outside $Z_r$ decreases exponentially with the expected number of neighbors $\lambda\pi r^2$. Theorem~\ref{thm:general_inZ} tells us that the expected number of sensors in $Z_r$ grows with the parameters $\lambda$ and $r$, and the perimeter of $X$. Thus a region $X$ with long boundary or with many components would be the worst in the majority vote scheme. Moreover such worst examples exist. As a result, Theorems~\ref{thm:general_outZ} and~\ref{thm:general_inZ} illustrate the trade-off between the error outside $Z_r$, which decreases exponentially with $\lambda$ and $r$, and the error inside $Z_r$, which increases with $\lambda$ and $r$.

\par
Sensors very near to the boundary of $X$ can be unavoidably misclassified according to Theorem~\ref{thm:general_inZ}. If $X$ is thin so that $X \subset Z_r$, then there are no sensors sufficiently deep inside $X$ whose neighbors are mainly inside $X$, so the region will not be recognized by the majority vote scheme. We thus need to make some (seemingly strong) assumptions on the shape of $X$ such that $X\not\subset Z_r$ is guaranteed. In this paper, we consider $X$ as a convex event region with a bounded curvature, i.e., with sufficiently rounded boundary. For such $X$, in Section~\ref{sec:convex}, we prove a significantly improved upper bound on the expected number of misclassified sensors in $Z_r$, which is a main result in this paper.

\begin{theorem}\label{thm:convex}
Let $p\leq 1/2$. If the event region $X$ is convex and the radius of curvature at each point on the boundary is at least $r$, then the expected number of sensors in $Z_r$ that are misclassified by the simple majority rule in the neighborhood of radius $r$ is less than
$${\pi\sqrt{\lambda}\over \sqrt{2}(1-2p)} \peri(X) + 3 \lambda\pi r^2 \ln{\peri(X)\over r}.$$
\end{theorem}

Finally, in Section~\ref{sec:experiment}, we perform some simulation for convex and round event regions and check the effect of the various parameters in the majority vote scheme such as $p$, $r$, $\lambda$, and the perimeter of $X$, and present a refinement method to improve the performance particularly for the tricky cases, i.e., for small $r$ and large $p$.

\section{Analysis for General Event Regions}~\label{sec:analysis}

In this section we analyze the simple majority rule and prove Theorem~\ref{thm:general_outZ} and Theorem~\ref{thm:general_inZ}. Throughout the paper, we will use the following lemma.

\begin{lemma}\label{lem:Bn}
For $p \leq 1/2$, the probability $B(n)$ of at least $\lceil{n\over 2}\rceil$ successes among $n$ independent Bernoulli trials of success probability $p$ is
$$ \frac{\sqrt{p(1-p)}}{2n}\left(2\sqrt{p(1-p)}\right)^{n} \leq B(n) \leq \left( 2\sqrt{p(1-p)} \right)^n.$$
\end{lemma}
\begin{proof}
The upper bound can be easily derived by the Chernoff inequality proven in~\cite{HR90}. For the lower bound, let $m=\lceil{n\over 2}\rceil$. Then $n \geq 2m-1$. By simple arithmetic calculation, we can show that the probability of at least $m=\lceil{n\over 2}\rceil$  successes among $n$ independent Bernoulli trials of success probability $p$ is at least
\[
\left( \begin{array}{c} 2m-1 \\ m \end{array} \right) p^m(1-p)^m
\geq \frac{1}{2} \left( \begin{array}{c} 2m \\ m \end{array} \right) p^m(1-p)^m
\geq \frac{2^{2m}}{4\sqrt{m}}\, p^m(1-p)^m .
\]
The last inequality is given in~\cite{MV08}.
Applying $n \leq 2m \leq n+1$ to the last term of the above inequality, we get the lower bound we wanted as follows:
\[
\frac{2^{2m}}{4\sqrt{m}}\, p^m(1-p)^m
\geq \frac{2^n}{2\sqrt{2n}} \left(\sqrt{p(1-p)}\right)^{n+1}
\geq \frac{\sqrt{p(1-p)}}{2n}\left(2\sqrt{p(1-p)}\right)^n.
\]
\end{proof}

\subsection{Proof of Theorem~\ref{thm:general_outZ}}

Recall that $Z_r$ is the set of points of $Y$ within distance $r$ to $\bd(X)$, the boundary of $X$. The expected number of the misclassified sensors in $Y \setminus Z_r$ by the majority vote is \( \lambda \area(Y \setminus Z_r) \) times the probability of a sensor $s$ in $Y\setminus Z_r$ being  misclassified by the majority rule in the neighborhood of radius $r$.

Suppose that $s$ has $k$ neighbors. Since $s \in Y\setminus Z_r$, the $k$ neighbors of $s$ lie all inside $X$ or all outside $X$. The probability of $s \in Y\setminus Z_r$ being misclassified is the one that at least half of measurements of the neighbors should be erroneous. When $k$ is odd, at least $\lceil{k\over 2}\rceil$ errors among $k$ measurements must happen. But when $k$ is even, the measurement of $s$ can be served as a tie breaker, thus at least $\lceil{k+1\over 2}\rceil$ errors must happen among $k+1$ measurements including a measurement of $s$.

Let $B(k)$ be the probability that at least $\lceil{k\over 2}\rceil$ successes among $k$ trials with success probability $p\leq 1/2$. For odd $k$, it holds from binomial distribution that $B(k)= {1\over 2(1-p)}B(k+1)\leq B(k+1)$ because ${1\over 2(1-p)}\leq 1$ for $p\leq 1/2$. The probability of $s \in Y\setminus Z_r$ being misclassified is simplified as follows:
\begin{eqnarray*}\label{eq:interior0}
 & & \sum_{k=0}^{\infty} \Pr(\hbox{$s$ has $k$ neighbors})\Pr\left({\hbox{$s$ makes a wrong decision by majority rule}}\right)\nonumber\\
 & = & \sum_{\mathrm{odd~} k} \Pr(\hbox{$s$ has $k$ neighbors}) B(k) +
       \sum_{\mathrm{even~} k} \Pr(\hbox{$s$ has $k$ neighbors}) B(k+1) \\
 & = &  \sum_{\mathrm{odd~} k} \Pr(\hbox{$s$ has $k$ neighbors}){1\over 2(1-p)}B(k+1)
        +  \sum_{\mathrm{even~} k} \Pr(\hbox{$s$ has $k$ neighbors}) B(k+1) \\
 &\leq& \sum_{k=0}^{\infty} \Pr(\hbox{$s$ has $k$ neighbors}) B(k+1).
\end{eqnarray*}

The first probability is \({1\over k!}(\lambda\pi r^2)^k e^{-\lambda\pi r^2}\) by the definition of the Poisson process. The second probability $B(k+1)$ is at most $\left(2\sqrt{p(1-p)}\right)^{k+1}$ by Lemma~\ref{lem:Bn}. Thus we get the upper bound of the probability of $s \in Y\setminus Z_r$ being misclassified as follows:
\begin{eqnarray*}
\Pr(\hbox{$s \in Y\setminus Z_r$ is misclassified})&\leq&
\sum_{k=0}^\infty {1\over k!}(\lambda\pi r^2)^k e^{-\lambda\pi r^2}
\left( 2\sqrt{p(1-p)}\right)^{k+1} \nonumber \\
&\leq& 2\sqrt{p(1-p)} e^{-\lambda\pi r^2} \sum_{k=0}^\infty {1\over k!}\left(\lambda\pi r^2 \cdot 2\sqrt{p(1-p)}\right)^k \label{eq:interior1}\\
&\leq& 2\sqrt{p(1-p)} e^{-(1-2\sqrt{p(1-p)})\lambda\pi r^2}.
\end{eqnarray*}
Multiplying this with $\lambda\area(Y\setminus Z_r)$ gives the upper bound of the theorem.

\par
For the lower bound in Theorem~\ref{thm:general_outZ} have a similar inequality as in the upper bound as follows:
\begin{eqnarray*}
 & & \sum_{\mathrm{odd~} k} \Pr(\hbox{$s$ has $k$ neighbors}) B(k) +
       \sum_{\mathrm{even~} k} \Pr(\hbox{$s$ has $k$ neighbors}) B(k+1) \nonumber \\
 & = &  \sum_{\mathrm{odd~} k} \Pr(\hbox{$s$ has $k$ neighbors}){1\over 2(1-p)}B(k+1)
        +  \sum_{\mathrm{even~} k} \Pr(\hbox{$s$ has $k$ neighbors}) B(k+1) \nonumber \\
 &\geq& \frac{1}{2}\sum_{k=0}^{\infty} \Pr(\hbox{$s$ has $k$ neighbors}) B(k+1),
\end{eqnarray*}
where ${1\over 2(1-p)}\geq \frac{1}{2}$ for $p\leq 1/2$.

The second probability is at least $\frac{\sqrt{p(1-p)}}{2(k+1)}\left(2\sqrt{p(1-p)}\right)^{k+1}$ by Lemma~\ref{lem:Bn}. Thus we get the following lower bound of the probability of $s \in Y\setminus Z_r$ being misclassified, which proves the lower bound of the theorem.
\begin{eqnarray*}
\Pr(\hbox{$s \in Y\setminus Z_r$ is misclassified})&\geq&
\frac{\sqrt{p(1-p)}}{4} \sum_{k=0}^\infty {1\over (k+1)!}(\lambda\pi r^2)^k e^{-\lambda\pi r^2}
\left( 2\sqrt{p(1-p)}\right)^{k+1} \\
&\geq& \frac{\sqrt{p(1-p)}e^{-\lambda\pi r^2}}{4\lambda\pi r^2} \sum_{k=0}^\infty {1\over (k+1)!}\left(\lambda\pi r^2 \cdot 2\sqrt{p(1-p)}\right)^{k+1} \label{eq:interior2}\\
&\geq& {\sqrt{p(1-p)} \over 4\lambda\pi r^2} \left(e^{-(1-2\sqrt{p(1-p)})\lambda\pi r^2}-e^{-\lambda\pi r^2}\right).
\end{eqnarray*}

\subsection{Proof of Theorem~\ref{thm:general_inZ}}
For the upper bound of the theorem, we simply assume that any sensor in $Z_r$ always makes the wrong decision. The expected number of sensors in $Z_r$ is $\lambda\area(Z_r)$, and we have the geometric bound $\area(Z_r)\leq 2r\peri(X) +\pi r^2\cdot\components(X)$ for any general set $X$. Thus the expected number of misclassified sensors in $Z_r$ by the majority vote rule is at most $\lambda\left(2r\peri(X) + \pi r^2\cdot\components(X) \right)$.

\begin{figure}
\centering
    \includegraphics[width=8cm]{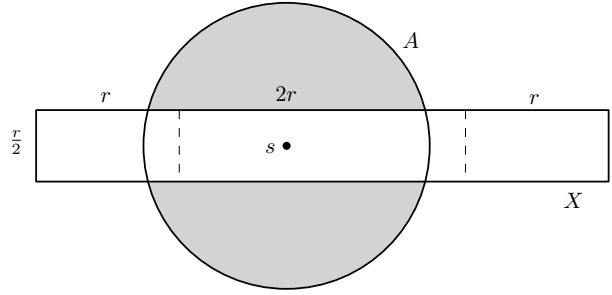}
\caption{A thin and long rectangle $X$.}
\label{fig:thin-lowerbound}
\end{figure}

\par
We now explain that this bound is asymptotically the best we can obtain for the expected number of misclassified sensors in $Z_r$ for any $X$ with $\components(X) = O (\peri(X)/r)$. Indeed, if $X$ is a thin and long rectangle of height $r/2$ and of width $4r$ as shown in Figure~\ref{fig:thin-lowerbound}, then all sensors in $X$ will be in $Z_r$, i.e., $X\subset Z_r$. The perimeter of $X$ is $9r$. Consider any sensor $s$ in $X$ at distance at least $r$ from the both vertical edges of $X$. Let $A$ be a disk of radius $r$ around $s$. Since the height of $X$ is $r/2$, $A$ consists of three parts as in Figure~\ref{fig:thin-lowerbound}; two circle segments of $A \setminus X$ and the middle part, $A\cap X$, between the circle segments. The expected numbers of sensors in $A\setminus X$ and $A\cap X$ whose initial measurement is ``not in X'' are $(1-p)\lambda\area(A\setminus X)$ and $p\lambda\area(A\cap X)$, respectively. Thus $s$ has at least $(1-p)\lambda\area(A\setminus X)+p\lambda\area(A\cap X)$ neighbors in $A$ whose initial measurement is ``not in $X$''. We can prove that this is at least half of the number of sensors in $A$, i.e., $\geq \frac{1}{2}\lambda\pi r^2$ for any $p \leq \frac{1}{2}$ by simple calculation. This results in making a wrong decision of $s$ by the majority vote scheme. The expected number of such misclassified sensors in $X$ is $\lambda r^2$, which is at least $\frac{1}{9}\lambda r(9r) = \frac{1}{9}\lambda r \peri(X) = \Omega(\lambda r \peri(X))$.

\begin{figure}
\centering
    \includegraphics[width=10cm]{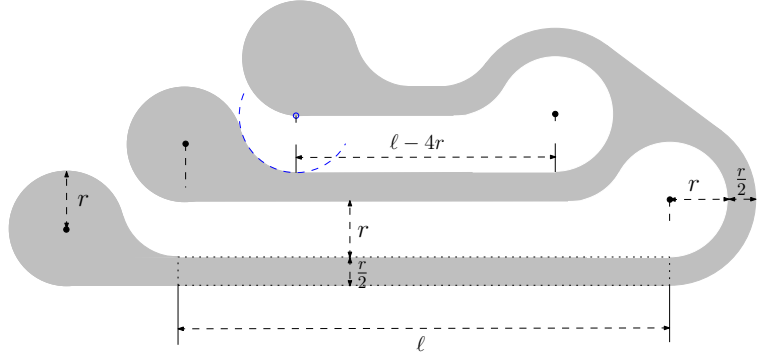}
\caption{A non-convex region $X$, satisfying that the radius of curvature is at least $r$ everywhere on the boundary.}
\label{fig:lowerbound}
\end{figure}

\par
We can also find such a worst example even when $X$ is not convex but has a round boundary satisfying some curvature constraint: the radius of curvature is at least $r$ everywhere on the boundary. Figure~\ref{fig:lowerbound} illustrates an event region $X$ of the curvature radius $r$, but not convex. All sensors in roughly $\lfloor \ell/4r \rfloor$ thin rectangular strips have a majority of neighbors outside $X$, thus they will make wrong decisions. Assume that $\ell$ is a multiple of $4r$, so $\ell \geq 4r$. Total area of thin strips is at least $(r/2)(\ell + (\ell-4r) + \ldots + 4r + 0) \geq \ell^2/16$. The boundary of $X$ consists of circular arcs at its both ends and linear segments of thin strips. Since the length of any circular arc is at most $3\pi r$, the total length of circular arcs is at most $6\pi r \ell/(4r)$. Then $\peri(X)\leq 6\pi r \ell/(4r) + 2(\ell + (\ell-4r) + \ldots + 4r)\leq (\ell^2/r)((3\pi/2 + 1)(r/\ell)+1/4) \leq 1.68\ell^2/r$ since $r/\ell \leq 1/4$. Thus the expected number of misclassified sensors in $Z_r$ of this non-convex region $X$ with bounded curvature $r$ is at least $\lambda \area(\mathrm{thin\,\, strips}) \geq \frac{1}{16}\lambda\ell^2 \geq \frac{1}{16}\lambda\left(\frac{r}{1.68}\peri(X)\right)\geq\frac{1}{27}\lambda r \peri(X) = \Omega(\lambda r \peri(X))$. We now complete the proof of Theorem~\ref{thm:general_inZ}.

\section{Analysis for Convex Event Regions with Round Boundary}~\label{sec:convex}

We now prove our main result, Theorem~\ref{thm:convex} that if $X$ is a convex region with a round boundary of the curvature radius $r$, then the expected number of misclassified sensors in $Z_r$ significantly decreases. As a result, the convexity and the curvature constraint both are crucial for a better bound.

Let $s$ be a sensor in $Z_r$ whose nearest point to $\bd(X)$ is $s'$ in distance $\delta r$ for some constant $\delta>0$. Let $A$ be the disc of radius $r$ around $s$. Let $\alpha$ be the fraction of $A$ on the same side of $\bd(X)$ as $s$, i.e., $\alpha:= \area(A\cap X)/\area(A)$ if $s\in X$, $\alpha := \area(A\setminus X)/\area(A)$ if $s \notin X$. Then we can prove the following:
\begin{lemma}\label{lem:good0}
If $p \le 1/2$ and $\alpha \ge 1/2$, then the probability of $s$ being misclassified is at most
\[  e^{-\frac{1}{2}\lambda\area(A)(1-2p)^2(2\alpha-1)^2}. \]
\end{lemma}
\begin{proof}
Assume first that $s$ lies in $Z_r \cap X$. Then $\alpha = \area(A\cap X)/ \area(A)$. The sensors are located in $A$ according to a Poisson distribution, but if we assume there are $k$ sensors in $A$,  which happens with the probability ${1\over k!}(\lambda \area(A))^k e^{-\lambda\area(A)}$, then the conditional distribution of these $k$ sensors over $A$ is uniform. So each of these $k$ sensors independently falls with probability $\alpha$ in $A \cap X$, and $1-\alpha$ in $A\setminus X$,
and there, with probability $p$, reports incorrectly, and with probability $1-p$, reports correctly.
Thus, we have $k$ independent Bernoulli experiments, which report being in $X$ with probability $\alpha(1-p)+(1-\alpha)p$ and being outside $X$ with probability $\alpha p +(1-\alpha)(1-p)$. Since $s$ is in $X$, the probability of an incorrect classification is the probability of a majority vote for being outside $X$, that is, more than half of $k$ Bernoulli experiments being successful with success probability $\alpha p +(1-\alpha)(1-p)$; if $\alpha p +(1-\alpha)(1-p)\leq {1\over 2}$, by the Chernoff bound (Lemma~\ref{lem:Bn}), this is at most
$$\left(
2\sqrt{(\alpha p +(1-\alpha)(1-p))
(\alpha(1-p)+(1-\alpha)p)} \right)^k.$$
For $p \le 1/2$ and $\alpha \ge 1/2$, the necessary condition for the Chernoff inequality, i.e., $\alpha p +(1-\alpha)(1-p)\leq{1\over 2}$, is satisfied, and thus the probability of $s$ being misclassified, without the condition on $k$, is at most
\begin{eqnarray}
&\phantom{=}&\sum_{k=0}^\infty
{1\over k!} (\lambda\area(A))^k e^{-\lambda\area(A)}
\left(
2\sqrt{(\alpha p +(1-\alpha)(1-p))
(\alpha(1-p)+(1-\alpha))} \right)^k \nonumber\\
&=& e^{-\lambda\area(A)(1-2\sqrt{(\alpha p +(1-\alpha)(1-p))
(\alpha(1-p)+(1-\alpha)p)})}~\label{eq:insideX}.
\end{eqnarray}
Set
\begin{eqnarray*}
g(\alpha, p) & = & 1-2\sqrt{(\alpha p +(1-\alpha)(1-p)) (\alpha(1-p)+(1-\alpha)p)}  \\
& = & 1-\sqrt{1-(1-2p)^2(2\alpha-1)^2}.
\end{eqnarray*}
For fixed $p\in [0, {1\over 2}]$, $g(\alpha,p)$ is a monotone
increasing function in $\alpha\in\lbrack{1\over 2},1\rbrack$ with $g({1\over2},p) = 0$ and $g(1,p) = 1-2\sqrt{p(1-p)}$, and satisfies that $g(\alpha,p) \ge {1\over2}(1-2p)^2(2\alpha-1)^2$. Thus we get the probability of $s$ in $Z_r \cap X$  being misclassified is at most $e^{-\frac{1}{2}\lambda\area(A)(1-2p)^2(2\alpha-1)^2}.$

\par
For sensors $s$ lying in $Z_r\setminus X$, letting $\alpha = \area(A \setminus X)/ \area(A)$, we get the same upper bound as (\ref{eq:insideX}) for the probability of $s$ being misclassified by more than half of neighbours reporting being inside $X$. Therefore we get the upper bound of the lemma for the probability of $s$ being misclassified, no matter whether $s$ is inside or outside $X$.
\end{proof}

\par
This lemma provides us a good bound on the expected number of misclassified sensors, but only for those satisfying its necessary condition, $\alpha \ge 1/2$. Since $X$ is convex, all sensors in $Z_r \setminus X$ satisfy the condition, but some sensors in $Z_r\cap X$ may not satisfy the condition. Indeed, for sensors $s \in Z_r\cap X$, we can get a simple lower bound on $\area(A \cap X)$ as follows: the sensor $s$ is assumed to be on the inner parallel curve at distance $\delta r$ from $\bd(X)$ for some $0 < \delta \leq 1$. Let $s'$ be the point on $\bd(X)$ from $s$ in distance $\delta r$. Consider another disk $B$ of radius $r$ with its center in $X$ such that it is tangent to $\bd(X)$ at $s'$. Since the radius of curvature is at least $r$ everywhere on $\bd(X)$, by Blaschke's rolling ball theorem~\cite[pp. 114-116]{BL}, the interior of $B$ is completely contained in $X$, thus we have $\area(A \cap X) \ge \area(A \cap B) = r^2 \beta(\delta)$, where
\[ \beta(\delta) = \frac{\pi}{2} + 2 \left[ \int^{1}_{\frac{1-\delta}{2}}\sqrt{1-x^2} dx - \int^{\frac{1-\delta}{2}}_0 \sqrt{1-x^2} dx \right] =  \frac{\pi}{2} + 2 \left[ \frac{\pi}{4}-2 \int^{\frac{1-\delta}{2}}_0 \sqrt{1-x^2} dx \right]. \]
Since $\beta(\delta)$ is a monotone increasing function with $\beta^{-1}(\frac{\pi}{2})=1-2\arcsin\frac{\pi}{8} > 0.2$, we can conclude that for sensors in $Z_r \cap X$ within distance  $0.2r$ to the boundary of $X$, it is not necessarily true that $\area(A\cap X) \ge \pi r^2 /2 = \area(A)/2$, i.e., $\alpha \ge 1/2$. Therefore it is unavoidable to split the sensors into ``good" and ``bad" groups and to analyze them in different ways.


\begin{figure}
\centering
    \includegraphics[width=12cm]{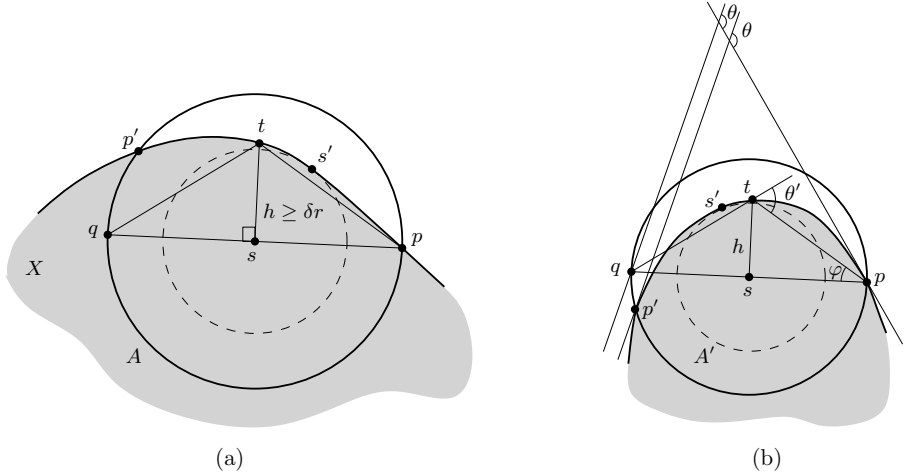}
\caption{Classification of sensors $s \in Z_r\cap X$. (a) $s$ is good since $q \in X$. (b) $s$ is bad since $q\not\in X$.}
\label{fig:bad}
\end{figure}

\par
For easier analysis, we use the following criteria for the split. Consider the disk $A$ of radius $r$ around $s$ and the boundary curve $\bd(X)$ as illustrated in Figure~\ref{fig:bad}. By Blaschke's rolling ball theorem, the curve $\bd(X)$ enters $A$ once at some point $p$, leaves $A$ once at some point $p'$ and it never enters $A$ afterwards. Let $q$ be the diametrically opposite point of $p$ in $A$. Then there are two possibilities when $s$ is in $Z_r$: $q$ lies in the same side of $\bd(X)$ as $s$, or in the opposite side of $\bd(X)$ to $s$. For the two cases when $s$ lies in $Z_r \cap X$, see Figure~\ref{fig:bad}. We call $s\in Z_r$ \emph{good} if $q$ lies in the side of $\bd(X)$ as $s$, and otherwise \emph{bad}. By the convexity of $X$, it is clear that all sensors in $Z_r \setminus X$ are good.

\subsection{Upper Bound on the Misclassified Good Sensors}

For good sensors in $Z_r$, we get the following bound on the probability of being misclassified.
\begin{lemma}\label{lem:good1}
Let $s$ be a good sensor in $Z_r$ whose distance to the boundary of $X$ is $\delta r$. Then the probability of $s$ being misclassified is at most
\[  e^{-\frac{2}{\pi^2}\lambda\area(A)(1-2p)^2\delta^2}. \]
\end{lemma}
\begin{proof}
As in Figure~\ref{fig:bad}(a), we first consider the case that $s\in Z_r\cap X$. Let $A$ be a disk of radius $r$ around $s$ on the inner parallel curve at distance $\delta r$ from $\bd(X)$, and let $s'$ be the closest point on $\bd(X)$ from $s$. Using the curvature constraint and the fact that $q\in X$, the half of $A$ bounded by line $pq$ is contained in $X$. In addition, on the other side of $pq$, the triangle $\triangle pqt$ of height $h\geq \delta r$, where $t$ is the point on $\bd(X)$ directly above $s$, is also contained in $X$. Thus $\area(A\cap X) \geq \frac{1}{2}\area(A)+\area(\triangle pqt) = \area(A)(\frac{1}{2}+\frac{\delta}{\pi})$. This gives us a lower bound on $\alpha(\delta)$, where the area of the portion of $A$ lying inside $X$ is expressed as $\alpha(\delta)\area(A)$, that is, $\alpha(\delta) = \area(A\cap X)/\area(A)$. Then $\alpha(\delta) \ge \frac{1}{2}+\frac{\delta}{\pi}.$
Similarly, for good sensors $s$ lying in $Z_r \setminus X$, the portion of $A$ lying outside $X$ satisfies that $\alpha(\delta) \ge \frac{1}{2}+\frac{\delta}{\pi}$.
Plugging the lower bound $\frac{1}{2}+\frac{\delta}{\pi}$ into $\alpha$ of Lemma~\ref{lem:good0}, we get the result.
\end{proof}

\par
We integrate this probability over all inner and outer parallel curves and get an upper bound on the expected number of misclassified good sensors in $Z_r$:
\begin{eqnarray}
&\phantom{\leq}& \hbox{Expected number of misclassified good sensors in $Z_r \cap X$} \nonumber\\
&& + \hbox{ Expected number of misclassified good sensors in $Z_r \setminus X$} \nonumber\\
&\leq &\lambda r \int_{0}^1 e^{-{2\over\pi^2}\lambda\area(A)
(1-2p)^2\delta^2}(\peri(X)-2\pi\delta r) d\delta \nonumber\\
& & + \lambda r \int_{0}^1 e^{-{2\over\pi^2}\lambda\area(A) (1-2p)^2\delta^2} (\peri(X)+2\pi\delta r) d\delta \nonumber\\
&= & 2 \lambda r \peri(X) \int_{0}^1 e^{-{2\over\pi^2}\lambda\area(A)
(1-2p)^2\delta^2} d\delta \nonumber\\
& < &  2 \lambda r \peri(X) {\sqrt{\pi}\over2\sqrt{ {2\over\pi^2}\lambda\area(A)
(1-2p)^2}} <{\pi\sqrt{\lambda}\over \sqrt{2}(1-2p)} \peri(X). \label{eq:good}
\end{eqnarray}
For the upper bound on the integrals, we used that $\int_0^1 e^{-cx^2}dx= \frac{\sqrt{\pi}}{2\sqrt{c}}\mathrm{erf}(\sqrt{c}) < \frac{\sqrt{\pi}} {2\sqrt{c}}$, where $\mathrm{erf}(x)$ is an error function appeared in integrating the Gaussian function with $\mathrm{erf}(x) < 1$ for any $x < \infty$. 

\subsection{Upper Bound on the Misclassified Bad Sensors}
Now we derive a bound on the expected number of misclassified bad sensors in $Z_r$. Note that all bad sensors of $Z_r$ appear only in $Z_r\cap X$, and see Figure~\ref{fig:bad}(b) for illustration. For bad points(sensors), we do not know any upper bound but $1$ on the probability of being misclassified. However, we can get an upper bound on the total length of disjoint bad curve segments, which consist of bad points only, on $C_\delta$, an inner parallel curve at distance $\delta r$ from $\bd(X)$ for some $0 < \delta \le 1$.
For this we use a covering argument and the fact that the total direction change of a simple closed convex curve is $2\pi$.

\begin{lemma}\label{lem:length_bad}
The total length of bad curve segments on the inner parallel curve $C_\delta$ at distance $\delta r$ to the boundary of $X$ is at most
\[ \min\left( {3\pi r\over \delta}, \peri(X)-2\pi\delta r \right).\]
\end{lemma}
\begin{proof}
As in the proof for the good sensors, we define $A$, $s'$, and $t$ for a bad sensor $s$ on the inner curve $C_\delta$. See Figure~\ref{fig:bad}(b). Then the disk $A'$ around $s$ of radius $\delta r$ touches $\bd(X)$ at $s'$ and is completely contained in $X$. Let $\theta$ be the angle by which the direction of $\bd(X)$ changes in counterclockwise direction between entering and leaving $A$. Using the facts that $h\ge \delta r$, $\theta \geq \theta'$, and $\arctan(x)\ge {\pi \over 4}x$ over $x\in [0, 1]$, we have that
$$\theta \geq \theta' = 2\varphi = 2\arctan(h/r) \geq 2\arctan(\delta)\ge {\pi\over 2}\delta.$$
Since the total direction change of $\bd(X)$ traversing a simple curve once around is at most $2\pi$, we cannot have more than ${4\over \delta}$ such curve segments on $C_\delta$ whose interiors are disjoint, each with a direction change of at least ${\pi\over 2}\delta$.

\begin{figure}
\centering
    \includegraphics[width=11cm]{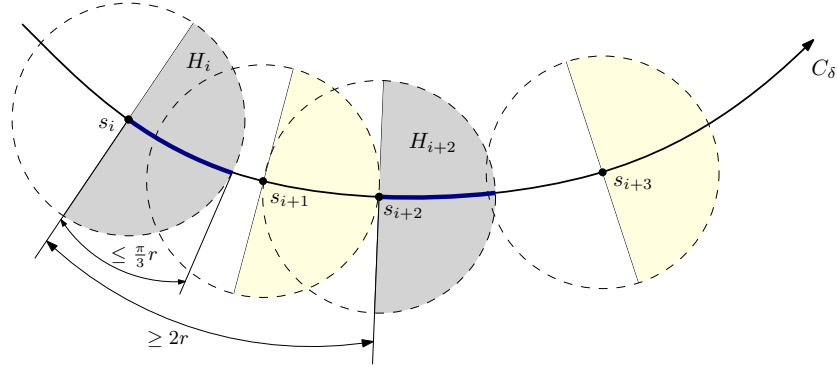}
\caption{Left half disks around each picked bad points.}
\label{fig:pick}
\end{figure}

We now pick bad points on $C_\delta$ at distance of at least $r$ along $C_\delta$ as traversing it in counterclockwise direction as follows. If the points on $C_\delta$ are all bad, then its length becomes the perimeter of $C_\delta$, i.e., at most $\peri(X)-2\pi\delta r$. Otherwise, there must be at least one good point on $C_\delta$. We traverse $C_\delta$ from the good point in counterclockwise direction. We will meet the first bad point, then we pick this bad point and call it $s_0$. We next pick the bad point $s_1$ on $C_\delta$ at distance of at least $r$ from $s_0$ along the curve. Continuing this picking process, we can pick $k$ bad points $s_0, s_1, \ldots, s_{k-1}$ where the distance from $s_{k-1}$ to $s_0$ might be less than $r$. As in Figure~\ref{fig:pick}(a), we denote by $H_i$ a right half of the disk of radius $r$ around each picked bad point $s_i$ with respect to the traversing direction. Then it is clear that the union of such right half disks covers all the bad points on $C_\delta$ because any bad point $s$ between $s_i$ and $s_{i+1}$\footnote{The addition on the indices is a modular addition with $k$.} is contained in $H_i$.

Without loss of generality, we assume that $k$ is odd. Let us now consider the intersections of $C_\delta$ with the right half disks $H_0, H_2, \ldots, H_{k-1}$ around every even picked bad points $s_0, s_2, \ldots, s_{k-1}$. These intersections $H_i \cap C_\delta$ for even $i$ result in curve segments (or arc intervals) of $C_\delta$ whose the left endpoint is $s_i$. We claim that these segments except from the first and last ones are disjoint; $H_{k-1}\cap C_\delta$ can overlap with $H_0\cap C_\delta$. As in Figure~\ref{fig:pick}, we consider two consecutive intersections, $H_i \cap C_\delta$ and $H_{i+2}\cap C_\delta$ for even $i < k-1$. It suffices to show that the right endpoint of $H_i\cap C_\delta$ lies in the left of the left endpoint of $H_{i+2}\cap C_\delta$ on the curve. The arc length between $s_i$ and $s_{i+2}$ is at least $2r$ by picking rule, and the length of $H_i\cap C_\delta$ is at most $\frac{\pi}{3}r$ by curvature constraint. Thus the distance between the right endpoint of $H_i\cap C_\delta$ and the left endpoint of $H_{i+2}\cap C_\delta$ is at least $2r-\frac{\pi}{3}r > 0$, so the claim is proved.

The number of disjoint bad curve segments on $C_\delta$ is already proved to be no more than ${4\over \delta}$, so the sum of their length (excluding the length of $H_{k-1}\cap C_\delta$) is at most $\frac{4}{\delta}\cdot \frac{\pi r}{3} = \frac{4\pi r}{3\delta}$. For the last curve segment $H_{k-1}\cap C_\delta$, we simply add its length $\frac{\pi r}{3}$, which gives the length of $\frac{(4+\delta)\pi r}{3\delta} \leq \frac{5\pi r}{3\delta}$ for $\delta \leq 1$. Considering the curve segments generated by every odd picked bad points, the sum of their length is at most $\frac{4\pi r}{3\delta}$. Note here that the first odd segment does not overlap with the last odd one. Thus the total length of bad curve segments on $C_\delta$ is at most $\frac{5\pi r}{3\delta}+\frac{4\pi r}{3\delta}\leq \frac{3\pi r}{\delta}$. Furthermore, the total length should be no more than the length of $C_\delta$, $\peri(X)-2\pi\delta r$, which completes the lemma.
\end{proof}

Integrating this over all inner parallel curves, we get an upper bound on the expected number of all misclassified bad sensors in $Z_r$ as follows:

\begin{eqnarray}
&\phantom{=}&\hbox{Expected number of misclassified bad sensors in $Z_r$} \nonumber\\
&\leq&\lambda r \int_{0}^1
\min\left({3\pi r\over \delta}, \peri(X)-2\pi\delta r \right)d\delta \leq
\lambda r \int_{0}^1 \min\left({3\pi r \over \delta}, \peri(X)\right)d\delta \nonumber\\
&=&\lambda r \left( \int_{0}^{3\pi r\over\peri(X)} \peri(X) d\delta
+ \int_{3\pi r\over\peri(X)}^1 {3\pi r \over \delta} d\delta \right) \nonumber\\
&=& 3 \lambda \pi r^2 \left( 1+ \ln {\peri(X)\over 3\pi r} \right) \leq  3 \lambda\pi r^2 \ln {\peri(X)\over r}. \label{eq:bad}
\end{eqnarray}

\par
Now we put both  (\ref{eq:good}) and (\ref{eq:bad}) together to obtain the upper bound on the expected number of misclassified points in $Z_r$, completing the proof of Theorem~\ref{thm:convex}:
\[{\pi\sqrt{\lambda}\over \sqrt{2}(1-2p)} \peri(X)+ 3 \lambda\pi r^2 \ln {\peri(X)\over r}.\]

\section{Simulation Results}\label{sec:experiment}

We perform some simulations to see the performance of the majority vote scheme and the dependency of several parameters such as the error probability $p$ of sensors, the neighboring radius $r$, and the geometric parameters of a convex event region $X$. We simulate the majority vote scheme each with $\lambda=2,500, 5,000, 10,000$, and $20,000$ sensors distributed by Poisson process in a unit square $Y$. We consider two event regions $X_S$ and $X_L$ as shown in Figure~\ref{fig:experiment}; $X_S$ is a square of side length $0.4$ with rounded corners of the curvature radius $0.1$, and $X_L$ is a longer and thinner rectangle of dimension $0.8\times 0.2$ with the same type of corners. Note here that they have the same area, but $X_L$ has a longer perimeter than $X_S$; $\peri(X_S)=0.2\pi+0.8$ and $\peri(X_L)=0.2\pi+1.2$. We test the majority vote scheme for $20$ different radii $r$ by incrementing $0.005$ from $0.005$ to $0.1$ and for $7$ different error probabilities $p$ by incrementing $0.05$ from $0.05$ to $0.35$.

\begin{figure}
\centering
\epsfig{file=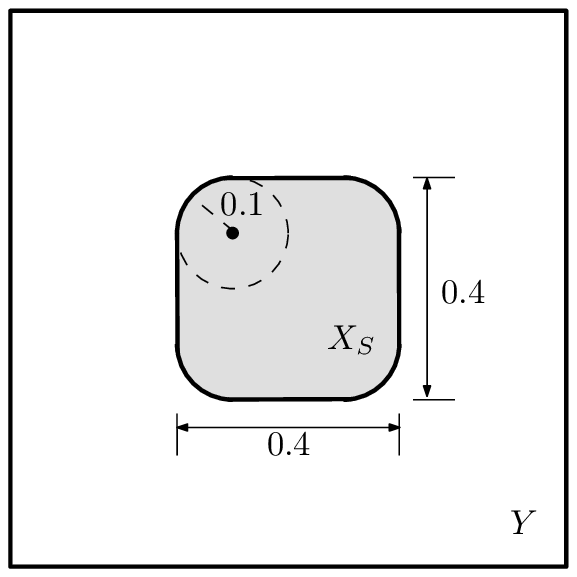, width = 4.5cm}\quad\quad\quad\quad
\epsfig{file=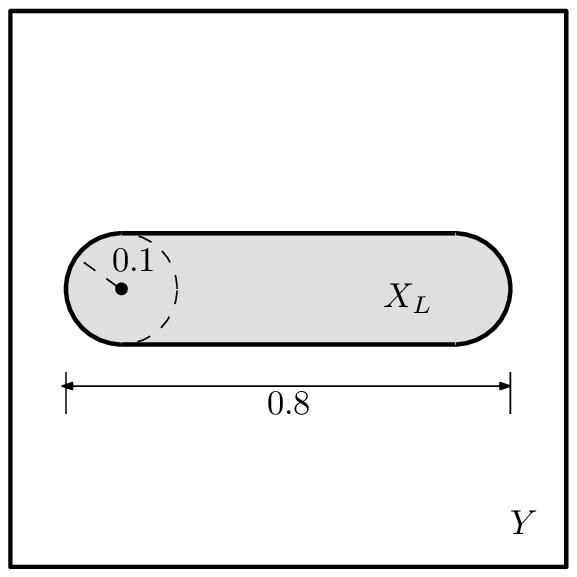, width = 4.5cm }
\caption{Event regions $X_S$ and $X_L$, having the same area but different perimeter.}
\label{fig:experiment}
\end{figure}

\subsection{Results}

We first check the correction rate of the major vote scheme, which is the ratio of the number of sensors revised correctly after the major vote scheme with the number of the initial errors. Figure~\ref{fig:sim-ratio} shows such correction rates for various values of $r$, $p$ and $\lambda$, respectively. If $r$ is small or $\lambda$ is small, then each sensor has too few neighbors to correct its error by the majority vote scheme, and what is worse is that the initial right decision can be changed even when the sensor is far from the boundary of $X$. However, as $r$ grows or $\lambda$ increases, much more misclassified sensors are correctly revised by neighborhood comparisons; more than $80\%$ for $r \geq 0.025$ or $\lambda \geq 1,000$. The correction rate for $p$ becomes the highest around $p = 0.15$; almost $86\%$ for $X_S$ and $84.7\%$ for $X_L$. But, if $p$ is too small or too large, then the rate goes below $80\%$, still above $70\%$.

\par
Figure~\ref{fig:sim-ub-final} shows the comparison on the difference between the number of final errors in the simulation and the upper bound on the number of the final errors proved in this paper for $X_S$ and $X_L$. The upper bound is given as the sum of two upper bounds, each in $Y\setminus Z_r$ of Theorem~\ref{thm:general_outZ} and in $Z_r$ of Theorem~\ref{thm:convex}, which is at most
\begin{eqnarray}
 & & 2\lambda \sqrt{p(1-p)}e^{-(1-2\sqrt{p(1-p)})\lambda\pi r^2}\area(Y\setminus Z_r) + \nonumber \\
 & & {\pi\sqrt{\lambda}\over \sqrt{2}(1-2p)} \peri(X) + 3 \lambda\pi r^2 \ln{\peri(X)\over r}\label{eq:ub}.
\end{eqnarray}
The number of final errors in the simulation is clearly much less than the theoretical upper bound, which tells us the upper bound could be improved further.

\par
We next test how much the sensors in $Z_r$ are likely to be misclassified out of the total final errors in $Y$. Figure~\ref{fig:sim-beta} shows that as $r$ grows (or $\lambda$ increases), the misclassified sensors occur mostly in $Z_r$ since the errors in $Y\setminus Z_r$ decreases exponentially with the expected number of neighbors $\lambda \pi r^2$ by the bound~(\ref{eq:ub}), but the errors in $Z_r$ increases linearly with $\lambda \pi r^2$. However, the high probability $p$ causes many initial errors everywhere in $Y$, which makes the voting effect less powerful for the sensors in $Y\setminus Z_r$, thus the ratio decreases as $p$ increases.

We can also see the effect of the perimeter. Since $X_S$ and $X_L$ have the same area, the average number of sensors falling into them is almost equal, but $X_L$ has a $28\%$ longer perimeter than $X_S$. A longer perimeter is a bad factor to misclassify more sensors outside the event region. According to Theorem~\ref{thm:convex}, we can expect there would be a linear relation between the number of misclassified sensors and the perimeter of $X$. We can find such relation in Figure~\ref{fig:sim-ratio} to Figure~\ref{fig:sim-beta}.

\par
We now count the misclassified sensors $Z_r\cap X$ and $Z_r \setminus X$ separately. We expected the misclassified sensors in $Z_r\cap X$ are more likely to occur than the ones in $Z_r \cap X$ since the sensors in $Z_r\setminus X_S$ are all good, i.e., $\area(A\setminus X) \geq \area(A)/2$, where $A$ is the disk of radius $r$ around a sensor in $Z_r \setminus X$. We have checked such situation actually happens as in Figure~\ref{fig:sim-gamma}.

\par
Finally, we conclude from the simulation results that the best radius $r$ that produces the least misclassified sensors is $0.02\leq r \leq 0.03$ for $p = 0.1$, $0.04\leq r \leq 0.05$ for $p = 0.2$, $0.07\leq r \leq0.08$ for $p = 0.3$, and $0.09\leq r \leq 0.1$ for $p= 0.4$.

\subsection{Further Refinements}

A sensor with a small neighborhood radius $r$, say less than $0.03$ in our simulation, has few neighbors to correct its error by the majority vote scheme, so the correction rate is not satisfactory as we already checked in Figure~\ref{fig:sim-ratio}. A refinement method to achieve the high correction rate for small $r$ is to apply the majority vote scheme more than once, i.e., multiple vote rounds. Suppose that an initial error of a sensor $s$ is not corrected during the first vote round. After the first vote, several erroneous neighbors of $s$ would be revised correctly, thus $s$ is more likely to correct its error at the second vote round. After $t\geq 1$ rounds, the measurement of the sensors at distance $tr$ from $s$ can affect the decision of $s$. This, on the other hand, tells that more rounds are not always helpful, particularly for the sensor $s$ near $\bd(X)$ in the sense that $s$ can receive the information from many sensors on the opposite side of $s$, which can lead $s$ to make the wrong change. Thus we need to choose $t$ carefully. For a fixed distance from $s$, as $r$ gets smaller, $t$ becomes larger, so $t$ should be set in proportional to $1/r$. For large error probability $p$, we can expect the similar effect as for small $r$, but its impact seems a bit weaker than that of $r$. We set $t = \max(cp/{r},1)$ and $c = 0.5$. In our simulation, $t$ is at most $35$.

\par
In our simulation, instead of simply repeating a majority vote scheme $t$ times, we take an indirect but efficient implementation as follows: With each sensor $s$, we associate a real value $\mathrm{score}(s)$. Initially, $\mathrm{score}(s):= 1$ if the initial measurement of $s$ is that $s \in X$ and $\mathrm{score}(s) := -1$ if $s\not\in X$. At each round, $\mathrm{score}(s)$ is updated to be the average value of $\mathrm{score}(s')$ for the neighbors $s'$ of $s$. After the first round, if $\mathrm{score}(s) > 0$, then it means $s$ has the majority of its neighbors inside $X$, otherwise outside $X$, which is the exactly same judgement as the single-round vote scheme. The score values are propagated gradually like the pattern in the well-known Gau{\ss}-Seidel iterative method, thus after $t$ rounds, each sensor $s$ receives the score values of the sensors within distance $tr$ from $s$. Each sensor $s$ decides its final measurement by the sign of $\mathrm{score}(s)$, that is, decides that $s\in X$ if $\mathrm{score}(s) >0$ and $s\not\in X$ if $\mathrm{score}(s) <0$. If $\mathrm{score}(s)=0$, then it follows the measurement of $s$ made at the previous round.

\par
Figure~\ref{fig:sim-s-m} compares the correction rates for $X_S$ and $X_L$ between single-round and multiple-round vote schemes. The multiple-round scheme performs much better for large $p \geq 0.3$; $11.3\%$ and $12.9\%$ improvements each for $p = 0.3$ and $p= 0.4$. If we focus only on small radius $r$ between $0.01$ and $0.03$, then the correction rates are improved more as in Figure~\ref{fig:sim-s-m}. As a result, the multiple-round vote scheme is a reasonable refinement to improve the correction rate for small $r$ and large $p$.

The simulation is done in the desktop computer with Intel Core i7-2600 CPU, 3.40GHz. Regardless of the values of $r$, $p$, and $\lambda$, every major vote scheme runs in $46$ seconds, which is the multiple-round case with $35$ rounds where $r = 0.005$, $p = 0.35$, and $\lambda = 20,000$, but the average execution time over all $r$, $p$, and $\lambda$ is $2.7$ seconds.

\begin{figure}
\centering
\mbox{\epsfig{file=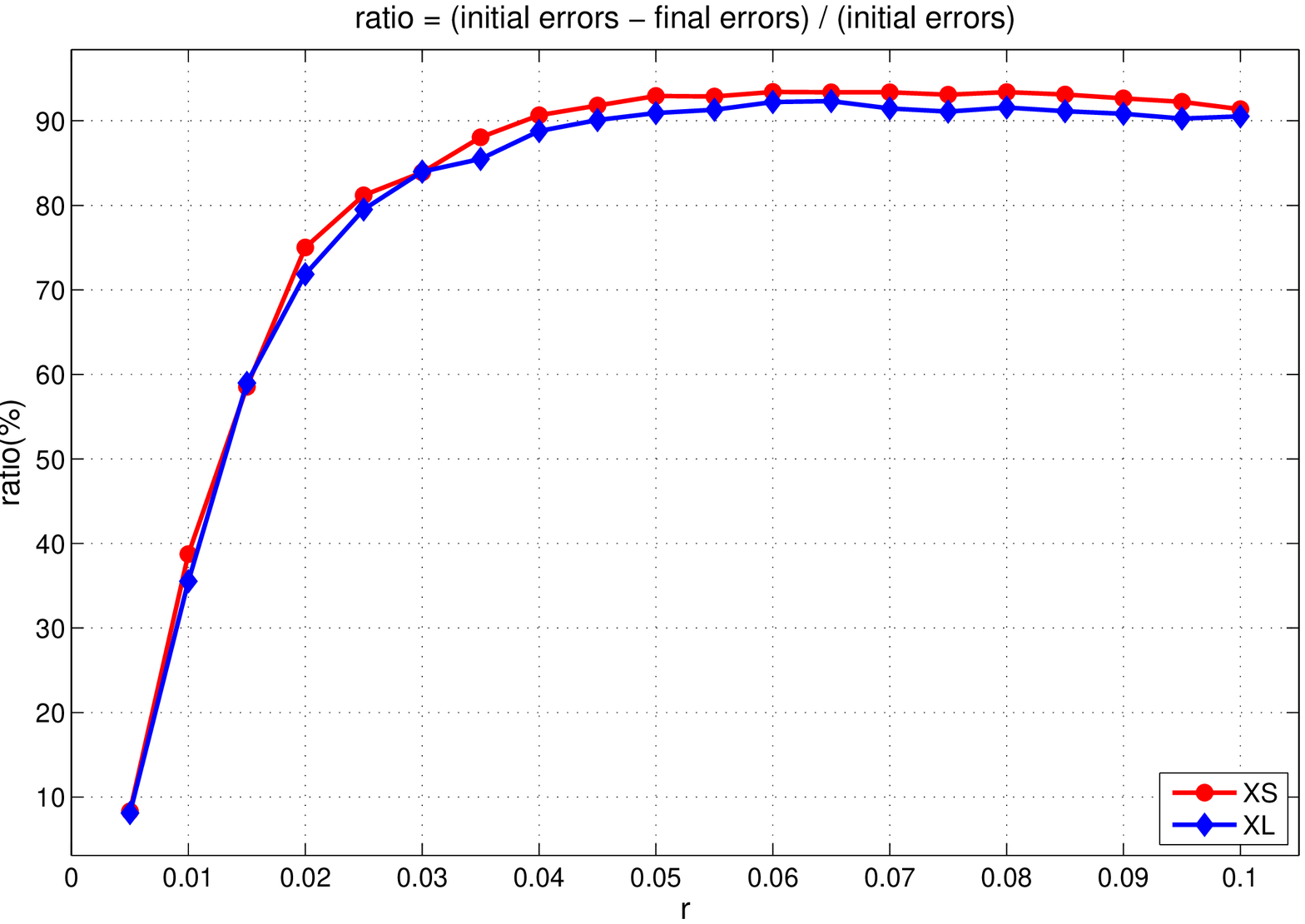, width=10cm}}
\mbox{\epsfig{file=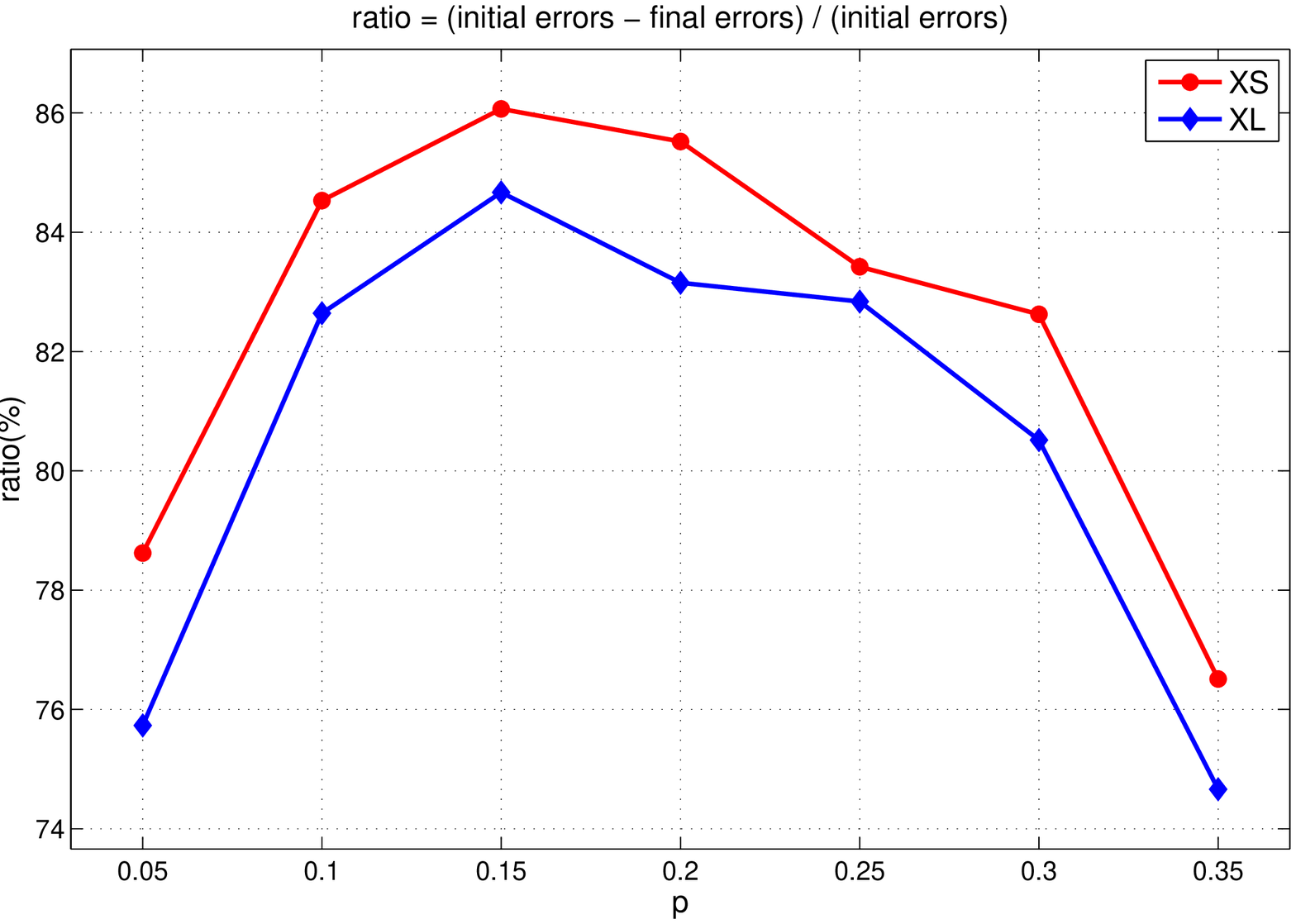, width=10cm}}
\mbox{\epsfig{file=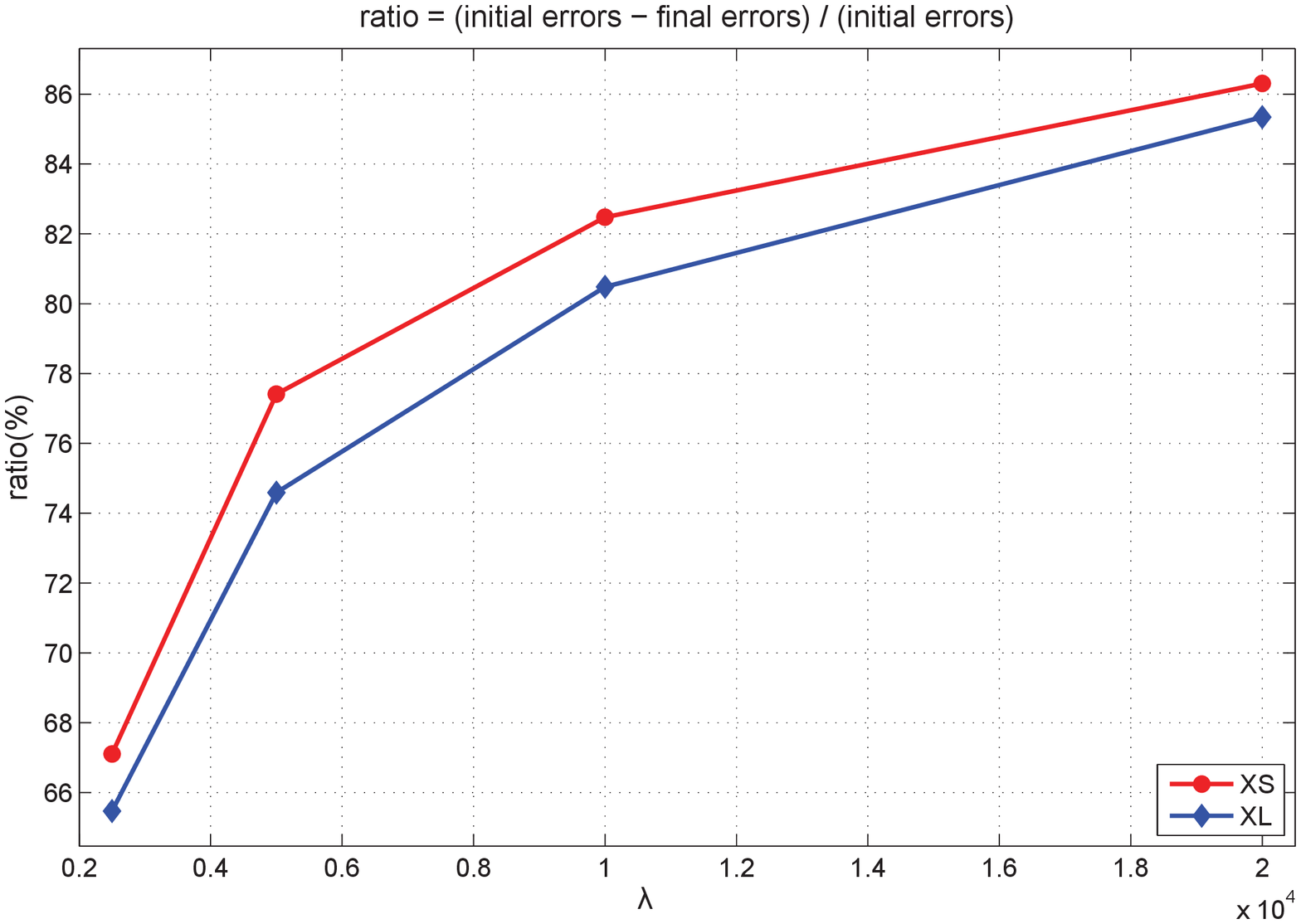, width=10cm}}
\caption{The correction rate by majority vote scheme for $X_S$ and $X_L$. (a) For $r$ (averaged by $p$ and $\lambda$). (b) For $p$ (averaged by $r$ and $\lambda$). (c) For $\lambda$ (averaged by $r$ and $p$).}
\label{fig:sim-ratio}
\end{figure}

\begin{figure}
\centering
\mbox{\epsfig{file=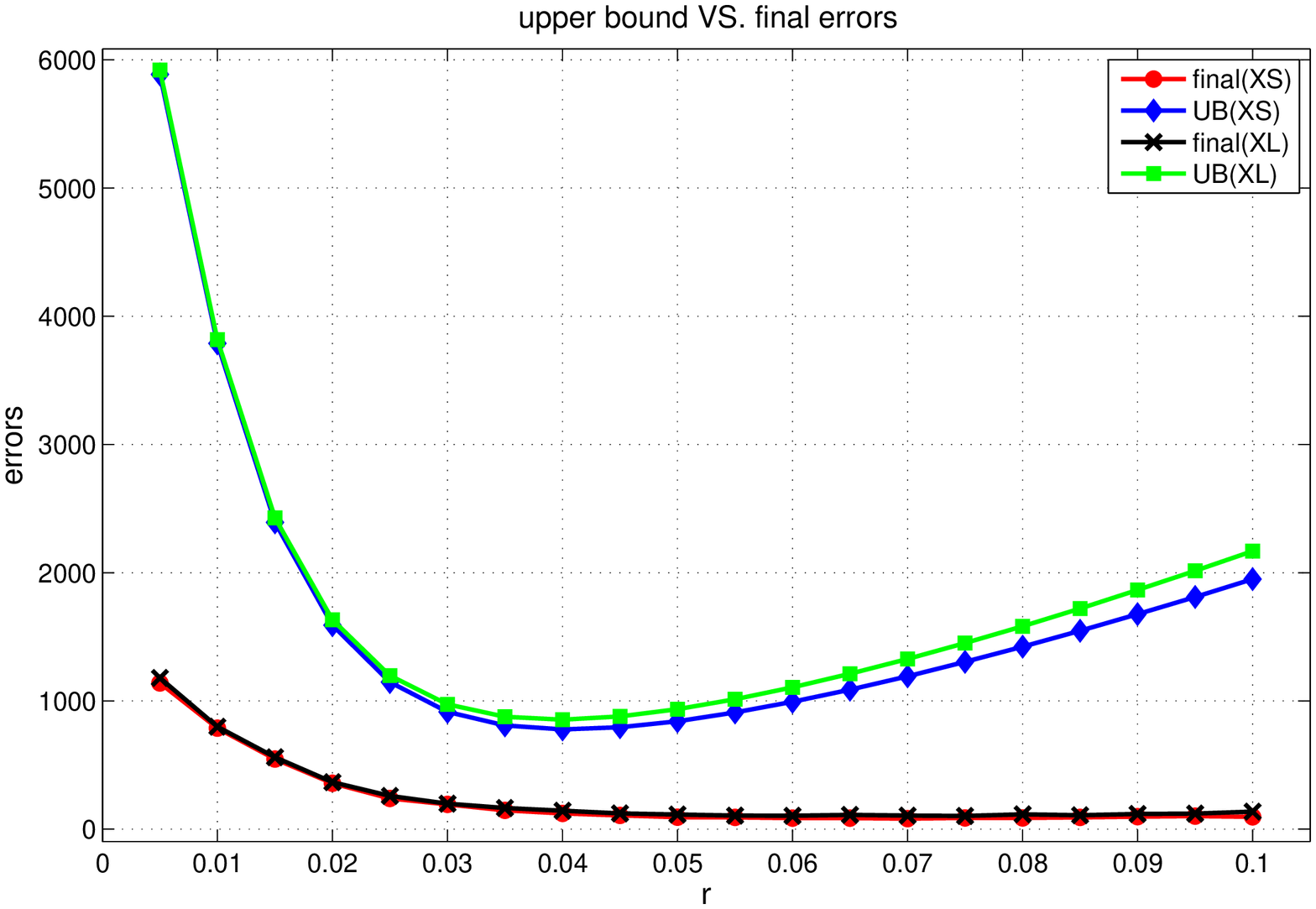, width=10cm}}
\mbox{\epsfig{file=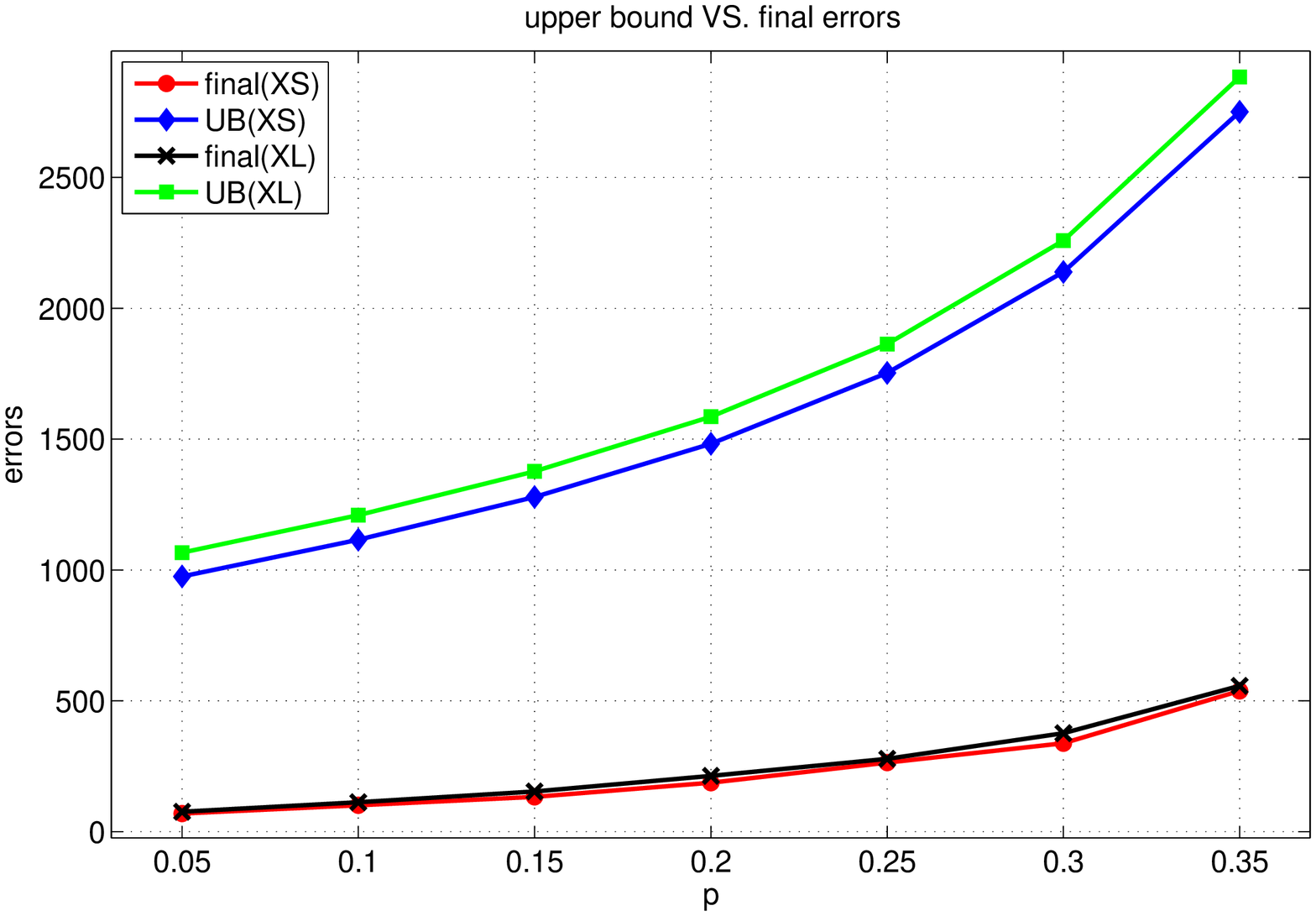, width=10cm}}
\mbox{\epsfig{file=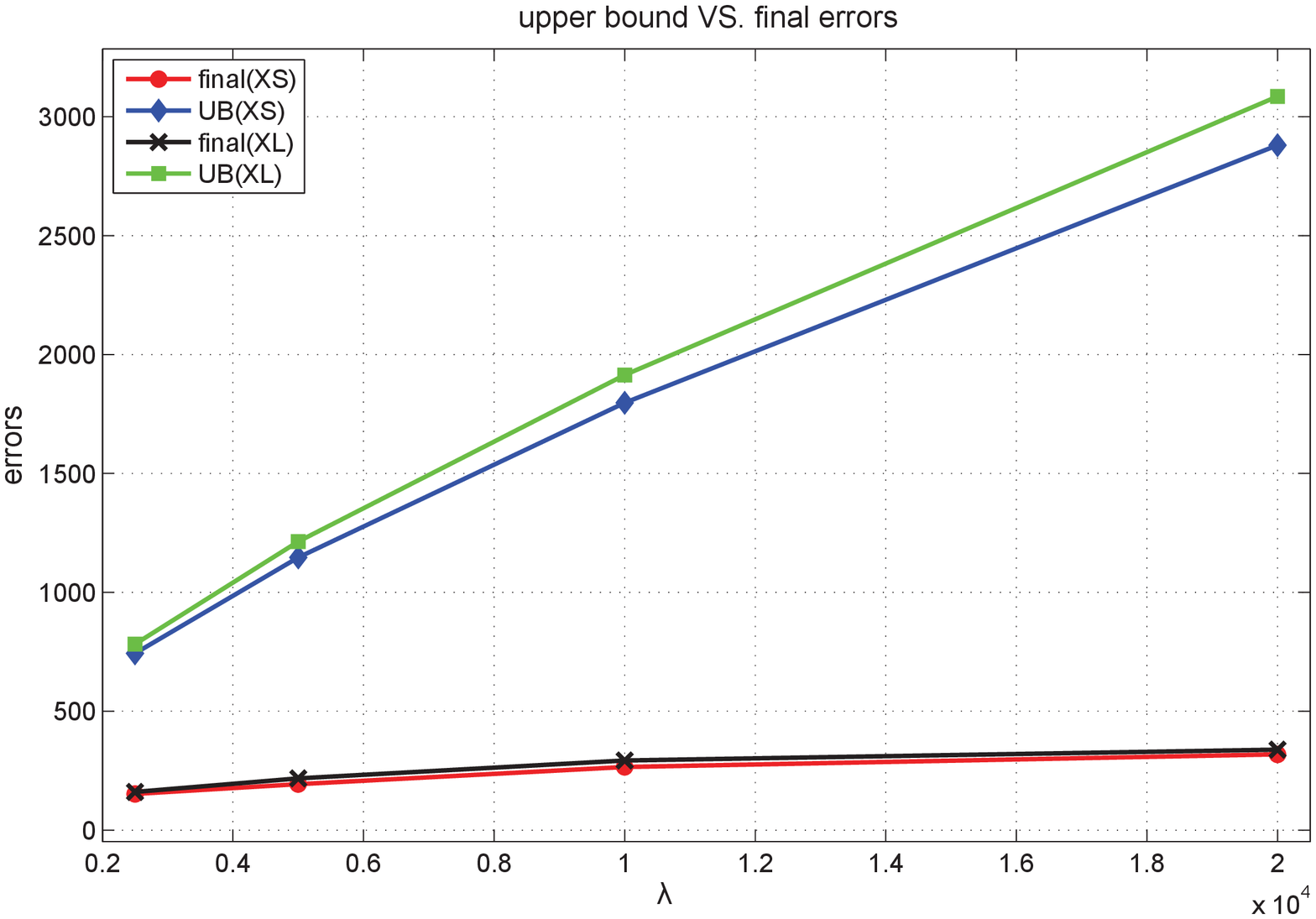, width=10cm}}
\caption{Comparisons between the number of final errors in simulations and the upper bound on the misclassified errors proved in Theorem~\ref{thm:general_outZ} and Theorem~\ref{thm:convex}.}
\label{fig:sim-ub-final}
\end{figure}

\begin{figure}
\centering
\mbox{\epsfig{file=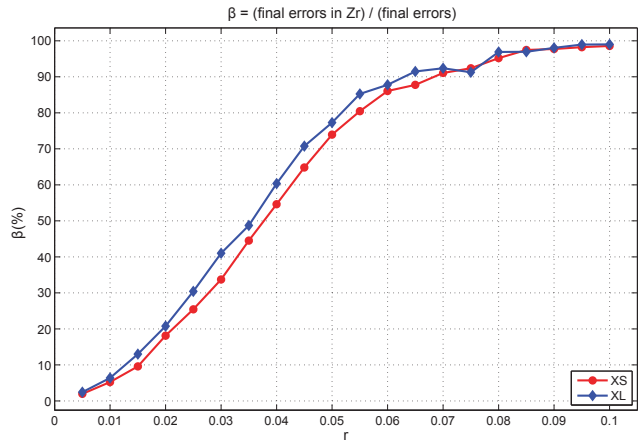, width=10cm}}
\mbox{\epsfig{file=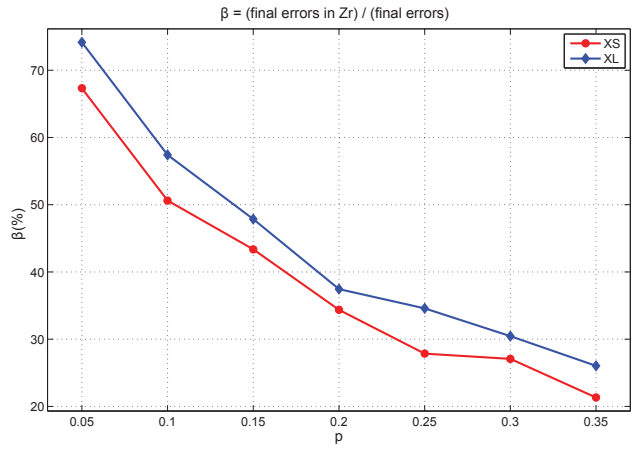, width=10cm}}
\mbox{\epsfig{file=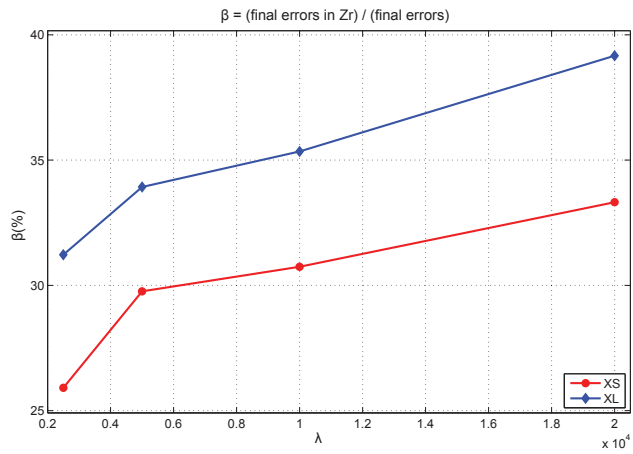, width=10cm}}
\caption{Ratio of misclassified sensors in $Z_r$ with the total errors in $Y$.}
\label{fig:sim-beta}
\end{figure}

\begin{figure}
\centering
\mbox{\epsfig{file=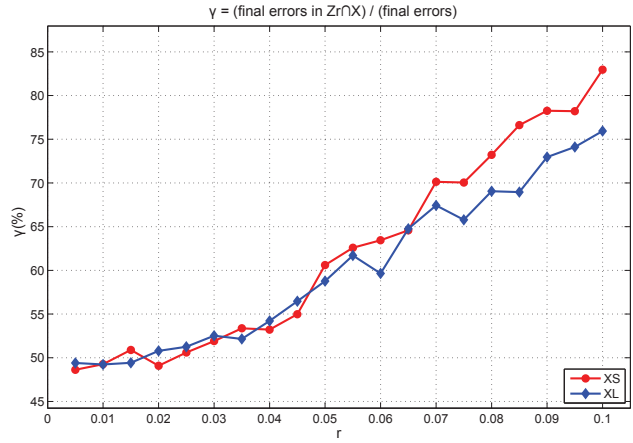, width=10cm}}
\mbox{\epsfig{file=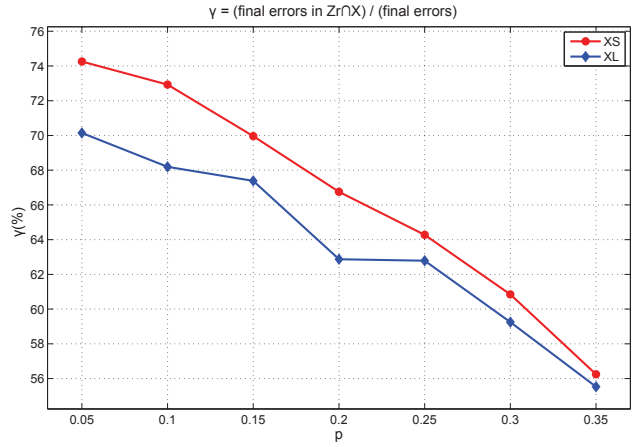, width=10cm}}
\mbox{\epsfig{file=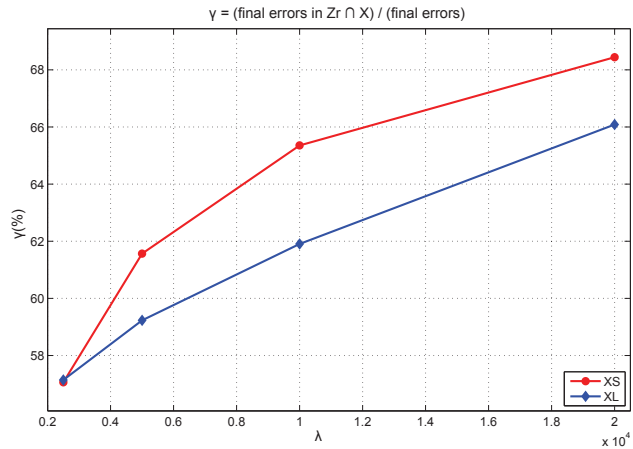, width=10cm}}
\caption{Ratio of errors in $Z_r\cap X$ with the errors in $Z_r$.}
\label{fig:sim-gamma}
\end{figure}

\begin{figure}
\centering
\mbox{\epsfig{file=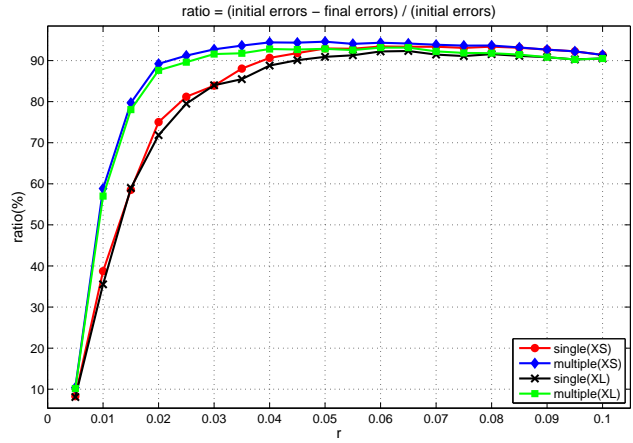, width=10cm}}
\mbox{\epsfig{file=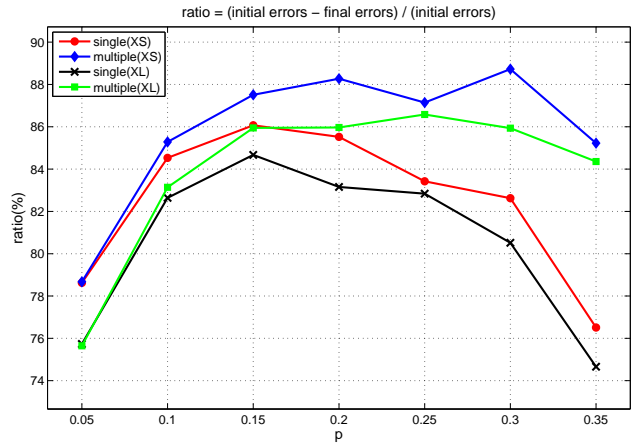, width=10cm}}
\mbox{\epsfig{file=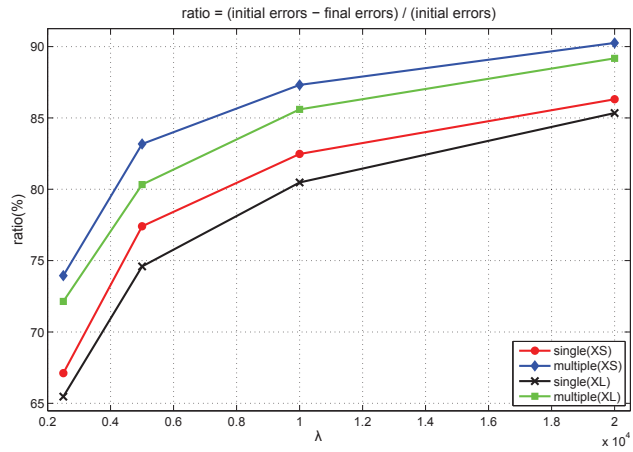, width=10cm}}
\caption{Comparison of single-round and multiple-round vote schemes.}
\label{fig:sim-s-m}
\end{figure}

\section{Conclusion}
In this paper we analyzed the simple majority rule and make explicit and precise the dependency on the error probability $p$ of sensors, the radius $r$ of the voting neighborhood, and the geometric parameters of event regions. To our best knowledge, this is the first to give bounds on the expected number of incorrectly classified sensors, with all such parameters in majority vote scheme. We also provided some empirical evidence indicating the dependency on such parameters.

The structure of our error bounds are the following:
\begin{itemize}
\item There is some background error, $2\lambda\sqrt{p(1-p)} e^{-(1-2\sqrt{p(1-p)})\lambda\pi r^2}\area(Y\setminus Z_r) $ which happens even when there is no event at all, i.e., $X=\emptyset$. This error depends on the total size of the area of interest $Y$, but decreases exponentially fast with the expected number of neighbors of each sensor, i.e., $\lambda\pi r^2$.

\item There is a term that depends on the perimeter of $X$, which means that sensors very near to the boundary of $X$ can be unavoidably misclassified. In fact, $\Omega(\lambda r \peri(X))$ sensors can be misclassified in a simplest thin rectangle of height $r/2$, which gives the lower bound for the term.

\item There are terms that depend on the expected number of neighbors of a point, the number of components of $X$, or the logarithm of the perimeter of $X$ specially for a convex region with bounded curvature.
\end{itemize}

\par
The assumption on the boundary curvature might look strong. We need it for two related things; for the existence of inner parallel curves of $\bd(X)$ at distance up to $r$, and for the property that any sensor neighborhood extends across $\bd(X)$ only on one side. The first might be a technical restriction which can somehow be circumvented, but the second is crucial for the voting algorithm: if the set $X$ is thin, then there are no sensor positions sufficiently deep inside $X$ that the majority of their neighbors will also be inside $X$, so the set will not be recognized by the majority rule.

\end{document}